\documentclass[a4paper, USenglish, cleveref, nameinlink, colorlinks=true, linkcolor=blue!75!green, urlcolor=blue!75!green, citecolor=blue!75!green, autoref]{lipics-v2021}

\usepackage{graphicx}
\usepackage[utf8]{inputenc}
\usepackage{xspace}
\usepackage{apptools}
\usepackage{mathtools}
\usepackage{orcidlink}
\usepackage{booktabs}
\usepackage[inline]{enumitem}
\usepackage{subcaption}
\usepackage[absolute,overlay]{textpos}
\usepackage{cite}

\nolinenumbers

\crefname{theorem}{Thm.}{Thms.}
\Crefname{theorem}{Theorem}{Theorems}
\crefname{corollary}{Cor.}{Cors.}
\Crefname{corollary}{Corollary}{Corollaries}

\graphicspath{{figures/}}

\newcounter{slope}
\crefname{slope}{}{}
\crefformat{slope}{#2#1#3}
\crefmultiformat{slope}{#2#1#3}{, #2#1#3}{, #2#1#3}{, #2#1#3}
\makeatletter
\newcommand{\slopeitem}[2]{%
  \item[$#1$]%
  \refstepcounter{slope}%
  \protected@edef\@currentlabel{$#1$}%
  \label{#2}%
}
\makeatother

\title{%
\texorpdfstring{%
    How Many Slopes Does Polynomial Area Cost?\\
    \smallskip
    \small The Slopebusters give the answer for planar graph drawings.}{%
How Many Slopes Does Polynomial Area Cost?
}}

\titlerunning{How Many Slopes Does Polynomial Area Cost?}

\author{Michael A. Bekos}{University of Ioannina, Greece \and \url{https://myweb.uoi.gr/bekos/} }{bekos@uoi.gr}{https://orcid.org/0000-0002-3414-7444}{}
\author{Eleni Katsanou}{National Technical University of Athens,  Greece} {ekatsanou@mail.ntua.gr}{https://orcid.org/0000-0002-1001-1411}{}
\author{Philipp Kindermann}{Trier University, Germany \and \url{https://algo.uni-trier.de/~kindermann}} {kindermann@uni-trier.de}{https://orcid.org/0000-0001-5764-7719}{}
\author{Maria Eleni Pavlidi}{University of Ioannina, Greece \and \url{https://algo.math.uoi.gr/marialena}}{m.e.pavlidi@uoi.gr}{https://orcid.org/0009-0009-4500-0112}{}

\authorrunning{M.\ A.\ Bekos, E.\ Katsanou, P.\ Kindermann, M.\ E.\ Pavlidi}

\Copyright{Michael A.\ Bekos, Eleni Katsanou, Philipp Kindermann, Maria Eleni Pavlidi}

\acknowledgements{This work was partially funded by the IKYDA 2025 funding program of the DAAD (project 57775422 ``Algorithms and Combinatorics for Drawing Graphs with Few Slopes'').}



\hideLIPIcs

\begin{document}

\maketitle

\begin{abstract}
\begin{textblock*}{10cm}(16cm,5.95cm)
\includegraphics[width=0.35\textwidth]{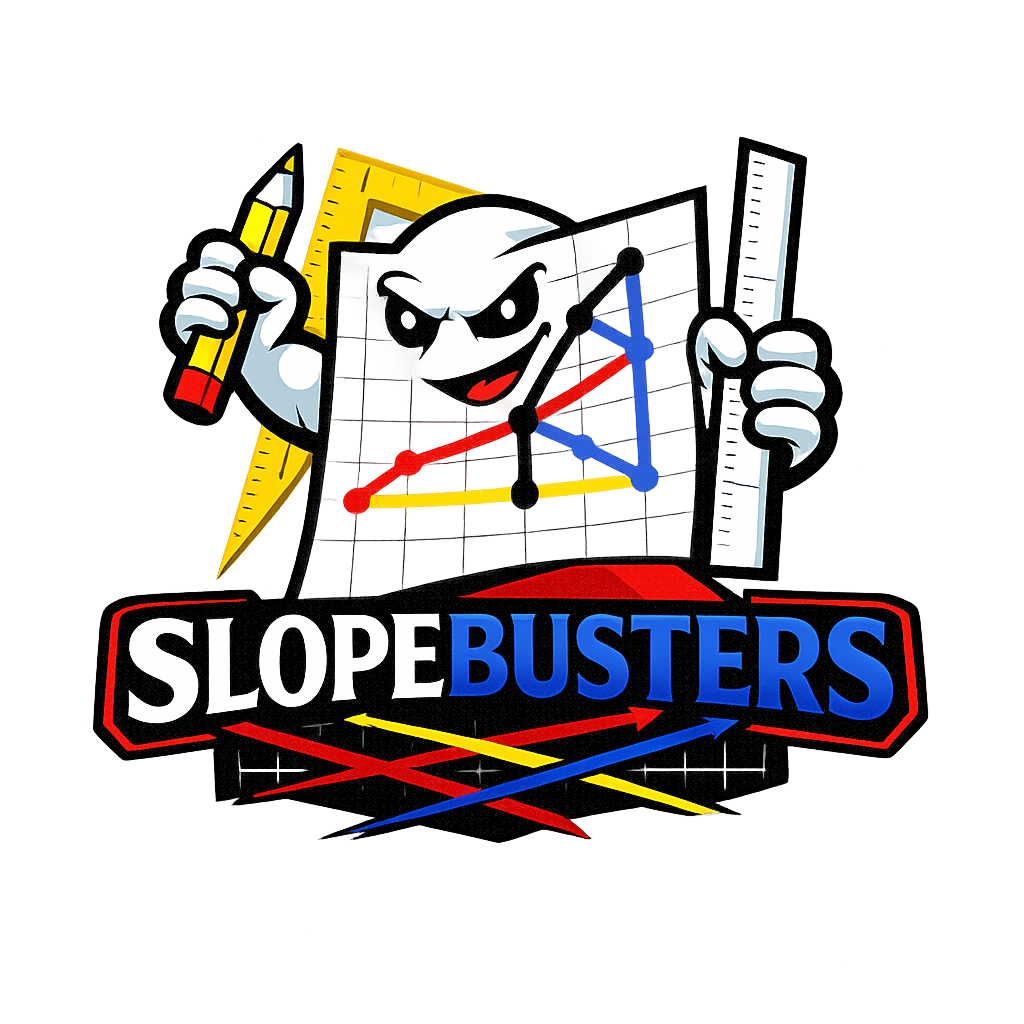}
\end{textblock*}%
In this work, we study the interplay between the number of slopes, the number of bends per edge, and the area requirements for planar drawings of bounded-degree graphs. Our motivation stems from the fact that, while numerous algorithms produce planar drawings with few slopes for graphs of relatively small degree in polynomial area, existing approaches for higher-degree graphs often require super-polynomial area. We address this gap in the literature by presenting new constructions that yield polynomial-area drawings with few bends per edge while slightly increasing the required number of slopes, thereby providing the first systematic study of slopes, bends and area trade-offs. 		
\end{abstract}

\section{Introduction}

Producing drawings of graphs while minimizing the number of slopes is a well-studied problem in Graph Drawing, motivated by both theoretical interest and practical applications such as network visualization and VLSI design \cite{BattistaETT99,Juenger04,2013gd}. 
From an algorithmic perspective, given a graph, the goal is to determine its slope number, that is, the minimum number of pairwise distinct slopes used by its edge segments, taken over all polyline drawings of the graph with a prescribed number of bends per edge.
Wade and Chu \cite{WadeChu1994} were among the first to investigate this problem by showing that the slope number of the complete graph $K_n$ in the straight-line setting is $n$. 
The problem has been extensively studied since then; several results are known for general graphs of bounded degree \cite{DujmovicSW07, KeszeghPPT08, MukkamalaP11, MukkamalaSzegedy2009, PachP06}, and beyond-planar graphs \cite{GiacomoLM15, KindermannMSS21}. 

When the input graph is planar, the output drawing is additionally required to be crossing-free.
The problem of finding planar drawings with few slopes has also been extensively studied~\cite{BekosGDLM22, KlawitterM22, KlawitterZ23,ChaplickLGLM24, GiacomoLM18, GiacomoLM20, DujmovicESW07, JelinekJKLTV13, KPP2013, KnauerMW14, LenhartLMN23, BrucknerKM22}, as it traces back to orthogonal graph drawing~\cite{GargT01,CornelsenK12,Tamassia1987}, where edges are represented as polygonal chains consisting of alternating horizontal and vertical segments, that is, using only two slopes.

In this context, a central result by Biedl and Kant \cite{BiedlKant1998} guarantees that every planar graph of maximum degree~$4$, except for the octahedron, admits an orthogonal drawing on an $n\times n$ grid in which each edge has at most two bends. This result is tight in the sense that there exist planar graphs of maximum degree~$4$ that do not admit planar orthogonal drawings with one bend per edge~\cite{Tamassia1987}.

A natural extension of the orthogonal drawing model is the octilinear, which additionally supports diagonal segments at $\pm 45^{\circ}$, yielding a total of four slopes. 
In this model, every planar graph with maximum degree at most $3$ admits a bendless planar drawing~\cite{GiacomoLM18} on a $O(n)\times O(n)$ grid, while every planar graph with maximum degree at most $4$ (and $5$, respectively) admits a planar drawing with at most one bend per edge on a $O(n^2)\times O(n)$ grid (and a super-polynomial grid, respectively)~\cite{BekosGKK2015}. 

For graphs of higher degree, Keszegh, Pach, and Pálvölgyi \cite{KPP2013} extended the algorithm of Biedl and Kant \cite{BiedlKant1998} and showed that every planar graph of maximum degree $\Delta \geq 3$, with the exception of the octahedron, admits a planar drawing with at most two bends per edge using segments of at most $\lceil \Delta/2 \rceil$ distinct equidistant slopes; this bound on the number of slopes is clearly optimal.  
Improving previous related results~\cite{KPP2013,KnauerW16}, Angelini, Bekos, Liotta, and Montecchiani~\cite{AngeliniBLM19} demonstrated that every planar graph of maximum degree $\Delta \geq 4$ admits a planar drawing with at most one bend per edge using segments from any arbitrary set of $\Delta - 1$ pairwise distinct slopes.

In the straight-line setting, Keszegh, Pach, and Pálvölgyi~\cite{KPP2013} showed that every planar graph with maximum degree $\Delta$ admits a straight-line drawing using segments of $2^{O(\Delta)}$ distinct slopes.

\medskip\noindent\textbf{Our contribution:} In this work, we identify and close a critical gap in the literature. 
While almost all aforementioned algorithms for planar graphs of small fixed degree produce drawings of polynomial area~\cite{BekosGKK2015,BiedlKant1998,GargT01,GiacomoLM18,Tamassia1987} (and are therefore practically applicable), the corresponding algorithms that have been proposed for planar graphs of higher maximum degree require super-polynomial area~\cite{AngeliniBLM19,KPP2013,KnauerW16}, which limits significantly their practical applicability.  
Thus, in this paper, we study for first time the interplay between the number of slopes, the number of bends per edge, and the corresponding area requirements; see~\Cref{tb:results}. Our focus is on algorithms that produce planar drawings with few bends per edge on polynomial-size grids while only slightly increasing the number of slopes used.  More precisely: 

\begin{table}[t!]
\caption{Summary of our results for planar graph drawings with at most $k$ bends per edge.}
\label{tb:results}
\centering
\renewcommand{\arraystretch}{1.25}
\small
\begin{tabular}{@{} c c c r@{$\,\times\,$}l c r @{}}
\toprule
\textbf{$k$} &
\textbf{Degree} &
\textbf{Connectivity} &
\multicolumn{2}{c}{\textbf{Drawing area}} &
\textbf{No.\ of Slopes} &
\textbf{Ref.} \\
\midrule
1 & $\Delta\ge 5$ & 3 &
$\mathcal{O}(\Delta n^2)$&$\mathcal{O}(\Delta n^3)$ &
$3\Delta - 8$ & \cref{th:degree-delta-1-bend} \\

1 & $\Delta\ge 5$ & -- &
$\mathcal{O}(\Delta n^2)$&$\mathcal{O}(\Delta n^3)$ &
$\lceil \tfrac{9}{2}\Delta \rceil +1$  & \cref{th:degree-delta-1-bend2} \\

1 & 5 & 3 &
$\mathcal{O}(n^3)$&$\mathcal{O}(n^4)$ & 5 & \cref{thm:1-bend_5-degree}
 \\

2 & $\Delta\ge 3$ & 2 &
$\mathcal{O}(n)$&$\mathcal{O}(\Delta n^2)$ &
$\lceil \tfrac{\Delta}{2} \rceil$ & \cref{thm:biedl-kant-biconnected} \\

2 & $\Delta$ & -- &
$\mathcal{O}(n)$&$\mathcal{O}(\Delta n^2)$ &
$\lceil \tfrac{\Delta}{2} \rceil+1$ & \cref{cor:biedl-kant-simply} \\

4 & $\Delta$ & -- &
$\mathcal{O}(n)$&$\mathcal{O}(n)$ &
$\Delta$ & \cref{thm:kaufmann-wiese}  \\

\bottomrule
\end{tabular}
\end{table}

\begin{itemize}
    \item In \Cref{sec:1-bend}, we prove that every 3-connected $n$-vertex planar graph of maximum degree~$\Delta$ admits a  planar grid drawing with at most one bend per edge, using at most $3\Delta - 8$ slopes on a $O(\Delta n^2)~\times~O(\Delta n^3)$ grid. Our approach builds upon the incremental construction by Angelini, Bekos, Liotta and Montecchiani~\cite{AngeliniBLM19}, which produces such drawings on \textit{arbitrary} sets of $\Delta-1$ slopes. In contrast, we fix the slope set in advance and increase its size from $\Delta - 1$ to $3\Delta - 8$ in order to guarantee polynomial area. As a consequence, for general planar graphs (i.e., not necessarily $3$-connected) the number of slopes becomes $\lceil \frac{9}{2}\Delta\rceil + 1$.
    \item For the special case of planar graphs of maximum degree~$5$, 
    we decrease the number of required slopes from $3\Delta - 8=7$ to $5$ at the cost of increasing the drawing area by a factor of $O(n^2)$; see \Cref{thm:1-bend_5-degree}. Compared with the best-known algorithm in~\cite{BekosGKK2015}, our construction achieves polynomial-area while increasing the number of slopes by one.
    \item In \Cref{sec:2-bend}, we prove that every planar graph $G$ of maximum degree $\Delta \ge 3$ admits a 2-bend planar drawing on a $O(n) \times O(\Delta n^2)$ grid using at most $\lceil \Delta/2 \rceil$ slopes if $G$ is biconnected, and at most $\lceil \Delta/2 \rceil+1$ slopes otherwise. In contrast to the algorithm by Keszegh, Pach, and Pálvölgyi~\cite{KPP2013}, which uses \emph{equidistant} slopes to support rotations and scalings of biconnected components around cut vertices, our algorithm uses a different slope set to guarantee polynomial-drawing area. However, it requires one additional slope for general (i.e., non-biconnected) planar graphs, as it can rely neither on rotations (because of the non-equidistant slopes) nor on scaling (because of the area requirement).
    \item In \Cref{sec:4-bend}, we prove that, regardless of the maximum degree $\Delta$ of the input planar graph, quadratic area in the number of vertices of the graph suffices to obtain a planar drawing with at most $\Delta$ slopes, where each edge has at most four bends. If the graph is additionally subhamiltonian, then three bends suffice.
\end{itemize}
Note that a fundamental requirement of our grid drawings is that both vertices and edge bends lie on grid points. Furthermore, we stress that the algorithms that we present in this work can be implemented to run in time linear in the size of the input graphs.

\section{Preliminaries}\label{sec:preliminaries}

In this section, we introduce preliminary definitions and notation used throughout the paper. Unless stated otherwise, all graphs considered are simple and undirected. The \emph{degree} of a vertex is the number of its neighbors. A graph has \emph{maximum degree}~$\Delta$ if it contains a vertex of degree~$\Delta$ and no vertex of degree greater than~$\Delta$. A graph is \emph{connected} if every pair of vertices is joined by a path. More generally, for $k \ge 1$, a graph is \emph{$k$-connected} if the removal of any set of at most $k-1$ vertices leaves the graph connected. In particular, $2$- and $3$-connected graphs are also referred to as \emph{biconnected} and \emph{triconnected}, respectively.

A \emph{drawing} of a graph maps each vertex of the graph to a point of the Euclidean plane and each of its edges to a Jordan arc connecting its endpoints. A drawing is \emph{planar} if no two edges intersect except possibly at common endpoints. Such a drawing partitions the plane into connected regions called \emph{faces}; the unbounded one is the \emph{outer face}. A graph is \emph{planar} if it admits a planar drawing. A \emph{planar embedding} of a planar graph is an equivalence class of planar drawings that define the same set of faces and the same outer face. A planar drawing is \emph{$k$-bend} if each of its edges is a polygonal chain composed of at most $k+1$ straight-line segments. The point where two such segments meet is called a \emph{bend}. Unless otherwise specified, we consider \emph{grid drawings}, that is, drawings in which each vertex and each bend lies on a point of the Euclidean plane with integer coordinates. Given such a drawing $\Gamma$, we denote by $W(\Gamma)$ and by $H(\Gamma)$ the width and the height of the minimum rectangle enclosing $\Gamma$, respectively.

The \emph{slope} of a line measures its steepness and direction. 
It is defined as the ratio of the vertical change (rise) to the horizontal change (run) between any two points on the line. 
This equivalently corresponds to the tangent of the counter-clockwise angle through which a horizontal line must be rotated to coincide with the given line. 
The \emph{horizontal} (\emph{vertical}) \emph{slope} is the slope of a line parallel (perpendicular) to the $x$-axis. 
The slope of an edge segment is the slope of the line containing it. 
Given a set of slopes $S$, a $k$-bend planar drawing is said to be \emph{on $S$} if each of its edge segments has a slope belonging to $S$. 
For a vertex $v$ in a $k$-bend planar drawing, each slope $s$ of $S$ determines two distinct rays emanating from $v$ with slope~$s$, which we call \emph{ports}. 
If $s$ is the horizontal slope, these rays are called \emph{horizontal}. 
The upward (downward) directed ports are called \emph{top} (\emph{bottom}) ports. 
We say that a port~$\rho_v$ incident to $v$ is \emph{free} if no edge incident to $v$ is drawn along $\rho_v$; otherwise,~$\rho_v$ is \emph{occupied}. 

Given an edge $e$ at the outer face of a $k$-bend planar drawing $\Gamma$, a \emph{cut at edge $e$} is a strictly $y$-monotone curve that
\begin{enumerate*}[label=(\roman*)]
\item starts at a point on a horizontal segment of $e$,
\item ends at a point on a horizontal segment of an edge $e'$ incident to the outer face of $\Gamma$, with $e' \neq e$, and
\item intersects only horizontal segments of $\Gamma$; see \Cref{fig:schnyder}.
\end{enumerate*}
Such a cut allows to \emph{stretch}~$\Gamma$ horizontally by translating all vertices and edges on one side of the cut horizontally by an arbitrary distance $d>0$, thereby increasing the horizontal distance between the two resulting parts without introducing crossings. Since the stretching is purely horizontal, the slopes of all non-horizontal segments remain unchanged, while the lengths of the horizontal segments crossed by the cut increase. Hence, if $\Gamma$ is on $S$ before the stretching, it remains on $S$ afterward. When we say that we \emph{stretch an edge} $e$, we refer precisely to this operation.

Let $G$ be a $3$-connected $n$-vertex plane graph and let $\Pi = (P_0,\ldots,P_m)$ be a partition of its vertex set into paths such that $P_0=\{v_1,v_2\}$, $P_m=\{v_n\}$, the edges $(v_1,v_2)$ and $(v_1,v_n)$ exist and belong to the outer face of $G$. 
For $i=0,\ldots,m$, let $G_i$ be the subgraph induced by $P_0\cup\ldots \cup P_i$ and denote by $C_i$ the contour of $G_k$ defined as follows:  If $i=0$, then $C_0$ is the edge $(v_1,v_2)$ of $P_0$, while if $i >0$, then $C_i$ is the path from $v_1$ to $v_2$ obtained by removing $(v_1,v_2)$ from the cycle delimiting the outer face of $G_i$. We say that $\Pi$ is a \emph{canonical order}~\cite{FraysseixPP90,Kant96} of $G$ if for each $i=1,\ldots,m-1$ the following hold (see~\Cref{fig:schnyder}): %
\begin{enumerate*}[label=(\roman*)]
\item $G_i$ is biconnected, internally $3$-connected and embedded with $C_i \cup \{(v_1,v_2)\}$ as its outer face; 
\item all neighbors of $P_i$ in $G_{i-1}$ are on $C_{i-1}$; 
\item $P_i$ either consists of a single vertex (called \emph{singleton}), or the degree of each of its vertices is $2$ in $G_i$ (called \emph{chain});
\item every vertex in $P_i$ has at least one neighbor in $P_j$ with $j>i$.
\end{enumerate*}
A canonical order of a $3$-connected planar graph can be computed in linear time~\cite{Kant96}.

\begin{figure}
    \centering
    \includegraphics[page=2]{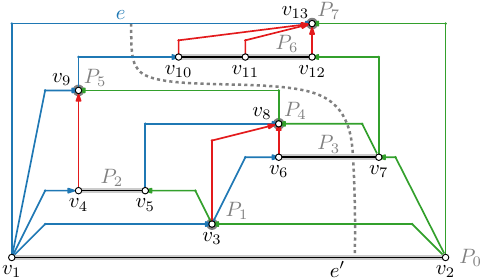}
    \caption{Illustration of a canonical order and of a $4$-coloring of a $3$-connected planar graph. The dotted curve is a cut (dashed) at edge $e$.}
    \label{fig:schnyder}
\end{figure}

Given a 3-connected plane graph $G$ and a canonical order $\Pi$ of it, it is possible to compute a 4-edge-coloring of $G$ similar to the one by Schnyder~\cite{Felsner04,Schnyder90}. The edge $(v_1,v_2)$ of~$G_0$ is colored black. For $i = 1,\ldots,m$, a 4-coloring of $G_{i-1}$ is extended to one of $G_i$ as follows (see, e.g, \Cref{fig:schnyder,fig:singleton}).
First, consider the edges of $G_i$ that do not belong to $G_{i-1}$ and lie on the contour $C_i$. The first (last) such edge encountered on a traversal of $C_i$ from $v_1$ to $v_2$ is colored blue (green, respectively), while all remaining ones (i.e., those whose endpoints both belong to $P_i$, when $P_i$ is a chain) are colored black. The remaining edges of $G_i$ that do not belong to $G_{i-1}$ are colored red; these are precisely the edges incident to~$P_i$ that are not part of $C_i$ (i.e., when $P_i$ is a singleton of the canonical order). 
Finally, we treat all black edges as undirected, and all remaining edges as directed where the orientation of an edge connecting a vertex $u \in P_i$ to a vertex $v \in P_j$ with $0 \leq i < j \leq m$ is from $u$ to $v$.

For an $n$-vertex graph with two designated vertices $s$ and $t$, an \emph{$st$-ordering} $v_1,\ldots, v_n$ is a permutation of its vertices such that $v_1=s$ and $v_n=t$, and every vertex $v_j$ with $1 < j < n$ has at least two neighbors $v_i$ and $v_k$ with $i<j<k$. It is known that every biconnected graph admits such an $st$-ordering~\cite{RosenstiehlT86}. In the case of a planar input, one may additionally guarantee that the vertices $v_1,v_2,v_n$ all lie on the outer face of the graph, such that the edge $(v_1,v_2)$ is incident to this face~\cite{BiedlKant1998}. 

In our algorithms, we occasionally augment the input graph to make it biconnected or triconnected. To augment a connected planar graph $G$ to a simply biconnected planar graph~$G'$, we utilize an algorithm by Kant and Bodlaender~\cite{KantBodlaender1991}. This algorithm runs in linear time and ensures that the maximum degree $\Delta(G')$ of the resulting graph satisfies $\Delta(G') \le \Delta(G) + 2$.
Furthermore, we utilize an algorithm by Kant~\cite{kant1993thesis} to augment a biconnected planar graph $G$ to a triconnected planar graph $G'$. This augmentation is performed in linear time, and the resulting maximum degree $\Delta(G')$ is bounded by $\Delta(G') \leq \max \{ 2, \lceil \frac{3}{2} \Delta(G) \rceil \}$.

\section{1-bend planar drawings of 3-connected planar graphs}\label{sec:general}

In this section, we consider $3$-connected planar graphs. We show how to obtain $1$-bend planar drawings with few slopes on a polynomially sized grid. We first present an algorithm for graphs with maximum degree $\Delta$, and afterwards we reduce the number of slopes for the special case of $\Delta = 5$. 

\subsection{\texorpdfstring{1-bend planar drawings of 3-connected degree-$\Delta$ planar graphs}{1-bend planar drawings of 3-connected degree-Δ planar graphs}}
\label{sec:1-bend}

We seek to prove that every 3-connected $n$-vertex planar graph with maximum degree~$\Delta$ admits a $1$-bend planar grid drawing with at most $3\Delta - 8$ slopes on a $O(\Delta n^2)~\times~O(\Delta n^3)$ grid. Since for $\Delta \leq 4$ our results is superseded by \cite{BekosGKK2015,DiGiacomoLM2014}, we assume w.l.o.g.\ that $\Delta \geq 5$. 
Our approach builds upon the incremental construction by Angelini, Bekos, Liotta and Montecchiani~\cite{AngeliniBLM19}, which uses a canonical ordering of the input $3$-connected planar graph $G$ to produce a $1$-bend planar drawing $\Gamma$ of it on an \textit{arbitrary} set of $\Delta-1$ slopes. In contrast, we fix the slope set in advance. Moreover, to obtain polynomial area, we enlarge the number of available slopes from $\Delta - 1$ to $3\Delta - 8$.

\begin{theorem}\label{th:degree-delta-1-bend}
Every 3-connected planar $n$-vertex graph $G$ with maximum degree $\Delta\ge 5$ admits a $1$-bend planar grid drawing with at most $3\Delta - 8$ slopes on a  $12\Delta n^2 \times 18\Delta n^3$ grid.
\end{theorem}

\begin{proof}
The slope set $S$ used by our algorithm is defined with respect to a parameter $k > \Delta n^2$ that we will specify later, and is the union of the following sets.
\begin{itemize*}
    \slopeitem{S_{\mathrm{v}}}{s:vert}consists only of the vertical slope.
    \slopeitem{S_{\mathrm{h}}}{s:hor}consists only of the horizontal slope.
    \slopeitem{S_{\mathrm{ls}}}{s:green}consists of $\Delta-3$ \emph{left steep} slopes $-\frac{k}{1},\ldots,-\frac{k}{\Delta-3}$ 
    (green in~\Cref{fig:slopes}).
    \slopeitem{S_{\mathrm{rs}}}{s:blue}consists of $\Delta-3$ \emph{right steep} slopes $\frac{k}{1},\ldots,\frac{k}{\Delta-3}$
    (blue in~\Cref{fig:slopes}). Finally,
    \slopeitem{S_{\mathrm{f}}}{s:red}consists of $\Delta-4$ \emph{flat} slopes 
    $\frac{1}{\Delta-3},\ldots,\frac{\Delta-4}{\Delta-3}$
    (red in~\Cref{fig:slopes}).
\end{itemize*}
Hence, the cardinality of $S$ is $3\Delta - 8$, as desired.

Let $\Pi = (P_0, \ldots, P_m)$ be a canonical order of $G$. For the drawing algorithm, we first remove the edge $(v_1,v_2)$ from the graph and color the edges incident to $v_1$ and $v_2$ in $G_1$ black. After the drawing algorithm is completed, we will insert the edge $(v_1,v_2)$ below the constructed drawing with a vertical segment incident to $v_1$ and a segment with slope $\frac{1}{\Delta-3}\in$\;\ref{s:red} incident to $v_2$ (see \Cref{fig:1-bend-drawing}).
For now, we also assume that $v_n$ has degree strictly less than $\Delta$ (note that if $\Delta\ge 6$, then this assumption is w.l.o.g.\, since we can choose $v_n$ such that its degree is at most $5$). We will describe how to handle the other case later.

Assume that for $0<i\leq m$, we have already constructed a $1$-bend planar grid drawing $\Gamma_{i-1}$ of $G_{i-1}$ on $S$ that satisfies the following invariants. 
\begin{enumerate}[label=I.\arabic*]     
    \item\label{inv:xmonotone} The contour $C_{i-1}$ of $\Gamma_{i-1}$ is drawn strictly $x$-monotone. 

    \item\label{inv:cut} There exists a cut at every edge belonging to the contour $C_{i-1}$ of $\Gamma_{i-1}$. 

    \item\label{inv:ports} Every vertex of $C_{i-1}$ has at least as many unoccupied top ports in each of \ref{s:blue} and \ref{s:green} incident to the outer face of $\Gamma_{i-1}$ as it has neighbors in $G \setminus G_{i-1}$ minus one, and the port corresponding to \ref{s:vert} is unoccupied if the vertex has at least one neighbor in $G \setminus G_{i-1}$ or its degree is strictly less than $\Delta$.
    
    \item\label{inv:ycoord} All vertices of $G_{i-1}$ are at $y$-coordinates that are multiples of $k$ in $\Gamma_{i-1}$.

    \item\label{inv:edge-colors} Based on their colors, the edges of $G_{i-1}$ have been drawn as follows in $\Gamma_{i-1}$:
    \begin{enumerate}[label=\alph*.,ref=I.5.\alph*] 
        \item\label{inv:black} Each black edge of $G_{i-1}$ consists of a single horizontal segment (i.e., its slope is in \ref{s:hor}).
        \item\label{inv:blue} Each blue edge of $G_{i-1}$ consists of two segments; the one incident to its source has a slope in \ref{s:vert}$\;\cup\;$\ref{s:blue}, while the one incident to its target is in~\ref{s:hor}. 
        \item \label{inv:red} Each red edge of $G_{i-1}$ consists of at most two segments; the one incident to its source is in \ref{s:vert}, while the one incident to its target has a slope in \ref{s:red}$\;\cup\;$\ref{s:vert}. If the slope of the second segment is in \ref{s:vert}, then the edge consists of one segment. 
        \item\label{inv:green} Each green edge of $G_{i-1}$ consists of two segments; the one incident to its source has a slope in \ref{s:vert}$\;\cup\;$\ref{s:green}, while the one incident to its target is in \ref{s:hor}. 
    \end{enumerate}    
\end{enumerate}

\begin{figure}[t]
    \centering
    \begin{subfigure}[t]{0.45\linewidth}
        \centering
    \includegraphics[page=1]{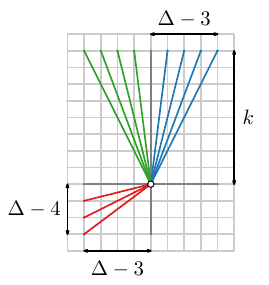}
       \caption{}
       \label{fig:slopes}
    \end{subfigure}\hfill
    \begin{subfigure}[t]{0.45\linewidth}
        \centering
        \includegraphics[page=3]{schnyder}
        \caption{}
        \label{fig:1-bend-drawing}
    \end{subfigure}
    \caption{(a) Illustration of the slopes used in \Cref{th:degree-delta-1-bend}. (b) The drawing created by the algorithm in \Cref{th:degree-delta-1-bend} with $k=4$ for the graph in \Cref{fig:schnyder}.}
\end{figure}

Note that Invariants~\ref{inv:xmonotone} and~\ref{inv:cut} are inherited from~\cite{AngeliniBLM19}; Invariant~\ref{inv:ports} is adapted to our setting, and Invariants~\ref{inv:ycoord} and~\ref{inv:edge-colors} are specific to our construction. 
The base case of our recursive algorithm is the graph $G_1$, which consists of the vertices $v_1$, $v_2$, and all vertices in $P_1$. We place these vertices on a horizontal line: $v_1$ is positioned at $(0,0)$, $v_2$ at $(|P_1|+1,0)$, and the vertices of $P_1$ are placed between them at unit distance. All edges are then drawn as horizontal segments. This obviously satisfies the invariants. 
So, we may assume that $i>1$. To derive a drawing $\Gamma_i$ of $G_i$ maintaining Invariants~\ref{inv:xmonotone}--\ref{inv:edge-colors}, we introduce the vertices of path $P_i$ into the drawing $\Gamma_{i-1}$ by distinguish two cases depending on whether $P_i$ is a chain or a singleton of degree 2 (Case 1) and a singleton of degree greater than 2 (Case 2). At a high level, Invariant \ref{inv:ycoord} will allow us to determine a grid point for each vertex of $P_i$ and for each bend of the blue and green edges incident to it, while Invariant \ref{inv:cut} will support stretching the drawing $\Gamma_{i-1}$ to create additional space to accomplish this placement, if needed. 

\subsection*{Case 1: $P_i$ is a chain or a singleton of degree~2}
Suppose that $P_i = \{v_g,\ldots,v_h\}$ is a chain or a singleton of degree $2$ in $G_{i}$. Note that in the latter case, $g=h$ holds. Let $v_\ell$ and $v_r$ be the neighbors of $v_g$ and $v_h$ in $G_{i-1}$, respectively. Refer to \Cref{fig:singleton-deg-2}. 
W.l.o.g., we may assume that $v_\ell$ appears before $v_r$ along the contour $C_{i-1}$ in a traversal of it from $v_1$ to $v_2$.
Let also $\rho_\ell$ ($\rho_r$) be the first unoccupied port at $v_\ell$ ($v_r$) encountered in a counter-clockwise (clockwise) traversal of its top ports in \ref{s:blue}\;$\cup$\;\ref{s:vert} (\ref{s:green}\;$\cup$\;\ref{s:vert}) starting from the rightward (leftward) horizontal port, which exist by Invariant~\ref{inv:ports}.
To satisfy Invariant~\ref{inv:ycoord}, we set the $y$-coordinate of each of $v_g,\ldots,v_h$ to $H(\Gamma_{i-1}) + k$. The $x$-coordinates of $v_g,\ldots,v_h$ will be determined by the constraints arising from the way that $(v_\ell, v_g), (v_g,v_{g+1}),\ldots, (v_{h-1},v_h), (v_h, v_r)$ must be drawn.
Since $(v_\ell, v_g)$ and $(v_r, v_h)$ are incoming edges to $v_g$ and $v_h$ in $G_i$, we satisfy Invariants~\ref{inv:blue} and \ref{inv:green} by drawing $(v_\ell, v_g)$ and $(v_r, v_h)$ with a horizontal segment incident to $v_g$ and $v_h$ and a second segment attached to $\rho_\ell$ and $\rho_r$, respectively, which maintains Invariant~\ref{inv:ports} for $v_\ell$ and $v_r$. The remaining edges of $P_i$ (if any) will be drawn as unit-length horizontal segments satisfying Invariant~\ref{inv:black}. Each of the vertices $v_g,\ldots,v_h$ has at most $\Delta-2$ neighbors in $G\setminus G_{i}$ and all $\Delta-3$ ports unoccupied in each of \ref{s:blue} and \ref{s:green} that are incident to the outer face of $\Gamma_i$, which ensures Invariant~\ref{inv:ports} for them.

\begin{figure}[t]
    \centering
    \begin{subfigure}[t]{0.45\linewidth}
        \centering
        \includegraphics[page=2,width=\linewidth]{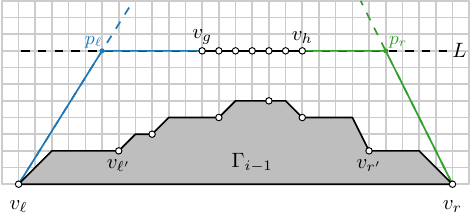}
        \caption{}
        \label{fig:singleton-deg-2}
    \end{subfigure}\hfill
    \begin{subfigure}[t]{0.45\linewidth}
        \centering
        \includegraphics[page=3,width=\linewidth]{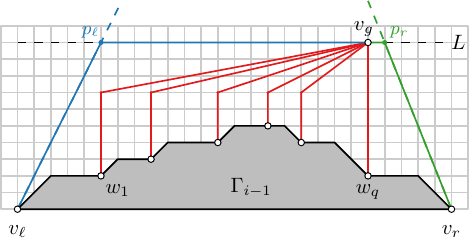}
        \caption{}
        \label{fig:singleton-deg-gt2}
    \end{subfigure}
    \caption{Illustration of the case where $P_i$ is (a)~a chain or a singleton of degree $2$, and (b)~a singleton of degree greater than $2$ in \Cref{th:degree-delta-1-bend}.}
    \label{fig:singleton}
\end{figure}
 
Let $p_\ell$ and $p_r$ be the points where the rays that correspond to $\rho_\ell$ and $\rho_r$ intersect the horizontal line $L$ through $v_g$; see \Cref{fig:badcases}. 
We have to ensure that $p_\ell$ lies at least $h-g+2$ units to the left of $p_r$ such that $v_g,\ldots,v_h$ can be placed between them in $\Gamma_i$. Furthermore, the rays at $\rho_\ell$ and $\rho_r$ are not allowed to cross any part of $\Gamma_{i-1}$. To guarantee these conditions, we horizontally stretch~$\Gamma_{i-1}$ as follows. 
Let $(v_\ell,v_{\ell'})$ be the first edge of $C_{i-1}$ that is encountered when traversing $C_{i-1}$ from $v_\ell$ to $v_2$. 
Analogously, $(v_r,v_{r'})$ is the first edge of $C_{i-1}$ that is encountered when traversing $C_{i-1}$ from $v_r$ to $v_1$.
(Note that $v_{\ell'}=v_r$ and $v_{r'}=v_\ell$ is possible.)
Since both these edges lie on $C_{i-1}$, by Invariant~\ref{inv:cut} there is a cut at each of them. 
This allows us to stretch $(v_\ell,v_{\ell'})$ until $p_\ell$ lies to the left of $v_{\ell'}$, and symmetrically, stretch $(v_r,v_{r'})$ until $p_r$ lies to the right of $v_{r'}$. 
This further guarantees that the horizontal distance between $p_\ell$ and $p_r$ is at least~$2$. If needed, we further stretch one of the edges $(v_{\ell},v_{\ell'})$ and $(v_{r'},v_{r})$, say w.l.o.g.\ the former, by up to $h-g$ additional units, so as to ensure that the horizontal distance between $p_\ell$ and $p_r$ is at least $h-g+2$. 
Hence, it is possible to place each of $v_g,\ldots,v_h$ at a grid point along $L$ that lies between $p_\ell$ and $p_r$, and draw the edges $(v_\ell,v_g)$ and $(v_r,v_h)$ with one bend each at $p_\ell$ and $p_r$, respectively. The remaining edges of $P_i$ are drawn as unit-length horizontal edge segments, as we initially sought; see~\Cref{fig:singleton}. This completes the drawing of~$\Gamma_i$.

It remains to show that $\Gamma_i$ is planar and satisfies Invariant~\ref{inv:xmonotone}. 
For the former, observe that the introduction of $v_g,\ldots,v_h$ into $\Gamma_{i-1}$ yields $h-g+4$ new edge segments in $\Gamma_i$. 
The two segments connecting $p_\ell$ with $v_g$ and $p_r$ with $v_h$, as well as all edge segments connecting internal vertices of $P_i$ (if any) lie completely above $\Gamma_{i-1}$. 
So they cannot cross any edge of $\Gamma_{i-1}$. The port $\rho_\ell$ that is used by the edge segment connecting $v_\ell$ with $p_\ell$ is the next available port in \ref{s:blue}, i.e., the one that follows the corresponding port used by the edge $(v_\ell,v_{\ell'})$. 
Furthermore, the only edge of $C_{i-1}$ that lies between the x-coordinates of $v_\ell$ and~$p_\ell$ is the edge $(v_\ell,v_{\ell'})$. 
So no edge of $C_{i-1}$ and thus no edge of $\Gamma_{i-1}$ can be crossed by the edge segment connecting $v_\ell$ and $p_\ell$. 
A symmetric argument applies to the edge segment connecting $v_r$ and $p_r$. Since these two edge segments cannot cross each other, it follows that~$\Gamma_i$ is planar, as desired. 
Having ensured this property, the fact that the contour $C_i$ of~$\Gamma_i$ is $x$-monotone is implied by the choice of $\rho_\ell$ and $\rho_r$ and by the fact that $\rho_\ell \in\;$\ref{s:blue} and $\rho_r\in\;$\ref{s:green}. 
Hence, $\Gamma_i$ satisfies Invariant~\ref{inv:xmonotone}. 

Before considering the case where $P_i$ is a singleton of degree greater than $2$ in $G_i$, we establish an upper bound on how much the drawing $\Gamma_{i-1}$ must be stretched horizontally to accommodate $v_g,\ldots,v_h$ in $\Gamma_i$.
To this end, consider the edge $(v_\ell,v_g)$. Let $x_\ell$ be the $x$-coordinate of $p_\ell$ before applying any stretching. Let also $s_\ell \in\;$\ref{s:blue} be the slope of the edge segment connecting $v_\ell$ with $p_\ell$; symmetrically, $s_r\in\;$\ref{s:green} is defined.  
It follows that $x_\ell=x(v_\ell)+\frac{y(v_g)-y(v_\ell)}{s_\ell}\le x(v_\ell)+\frac{y(v_g)}{s_\ell}$.
After the stretching of the edge $(v_\ell,v_{\ell'})$, the $x$-coordinate of $p_\ell$ should be at most $x(v_{\ell'})-1$. To achieve this,  the edge $(v_\ell,v_{\ell'})$ must be stretched by at most $x_\ell-(x(v_{\ell'})-1)\le x(v_\ell)+\frac{y(v_g)}{s_\ell}-x(v_{\ell}) =\frac{y(v_g)}{s_\ell}$ units of length. Symmetrically, the edge $(v_r,v_{r'})$ must be  stretched by at most $\frac{y(v_g)}{|s_r|}$ units of length.
Now the $x$-coordinate of $p_\ell$ is at most the old $x$-coordinate of $v_\ell$ and the $x$-coordinate of $p_r$ is at least the old $x$-coordinate of $v_r$, so $p_\ell$ and $p_r$ are at least one unit apart. To make space for $v_g,\ldots,v_h$ between them, we might have to stretch the edge $(v_\ell,v_{\ell'})$ by $h-g+1$ more units. Since $|P_i|=h-g+1$, it follows that 
the total stretch applied is at most

\begin{align}\label{eq:width-1}
\notag \frac{y(v_g)}{s_\ell} + \frac{y(v_g)}{|s_r|} + |P_i|&\le \frac{H(\Gamma_{i-1})+k}{\frac{k}{\Delta-3}} + \frac{H(\Gamma_{i-1})+k}{\frac{k}{\Delta-3}} + |P_i|\\
&=2\cdot \frac{H(\Gamma_{i-1})+k}{k}\cdot(\Delta-3)+|P_i|\le 2\Delta\cdot \frac{H(\Gamma_{i-1})+k}{k}+|P_i|.
\end{align}

\begin{figure}
    \centering
    \includegraphics[page=2]{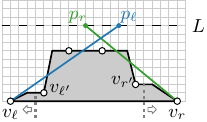}
    \includegraphics[page=3]{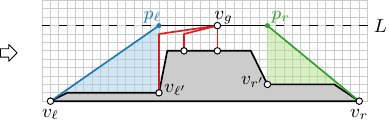}
    \caption{Illustration of using cuts to maintain planarity in \Cref{th:degree-delta-1-bend}.}
    \label{fig:badcases}
\end{figure}

\subsection*{Case 2: $P_i$ is a singleton of degree more than 2}
Suppose that $P_i = \{v_g\}$ is a singleton of degree greater than $2$ in $G_{i}$.
Let $v_\ell,w_1,\ldots,w_q,v_r$, with $q\geq 1$, be the neighbors of $v_g$ in $\Gamma_{i-1}$ as they appear from left to right  along $C_{i-1}$. 
We will place $v_g$ above~$w_q$, such that the edge connecting them is vertical; see \Cref{fig:singleton-deg-gt2}. Equivalently, this corresponds to setting the $x$-coordinate of $v_g$ to the one of $w_q$. The $y$-coordinates of $v_g$ will be determined by the constraints arising from the way that $(v_\ell, v_g), (v_g,w_1),\ldots, (v_g,w_q), (v_g, v_r)$ must be drawn (see Invariants~\ref{inv:blue}, \ref{inv:red} and \ref{inv:green}).
More precisely, each edge $(w_j,v_g),1\le j\le q$ will be drawn with a vertical segment incident to $w_j$, which is unoccupied by Invariant~\ref{inv:ports}, and, if $j \neq q$, with a second segment of slope $s_j=\frac{j}{\Delta-4}\in\text{\ref{s:red}}$. Clearly, if $q \leq \Delta-3$, then Invariant \ref{inv:red} is satisfied. This condition always holds when $i \neq m$, or when $i = m$ and the degree of $v_n$ is strictly less than $\Delta$. We will discuss the remaining case later.
Furthermore, the edges $(v_\ell,v_g)$ and $(v_g,v_r)$ will be drawn afterwards as in Case 1, thereby satisfying Invariants~\ref{inv:cut}, \ref{inv:ports}, \ref{inv:blue}, and \ref{inv:green}.

To ensure that the bend point of each edge $(w_j,v_g),1\le j< q$ lies on a grid point, we first ensure that the horizontal distance between $w_j$ and $v_g$ is a multiple of $\Delta-3$. To this end, let $(w_j,w'_j)$ be the first edge of $C_{i-1}$ that is encountered when traversing $C_{i-1}$ from $w_j$ to $v_2$ 
(possibly $w'_j=w_{j+1}$).
We leverage Invariant~\ref{inv:cut} to stretch $(w_j,w'_j)$ by at most $\Delta-4$ units so as to guarantee the required distance between $w_j$ and $v_g$. 
Overall, this increases the width by at most $(q-1) \cdot (\Delta-4)$ units.
To satisfy Invariant~\ref{inv:ycoord}, we set the $y$-coordinate of $v_g$ to $H(\Gamma_{i-1}) + \alpha k$, where $\alpha \in \mathbb{N^*}$ is a parameter that is  chosen such that, for every neighbor $w_j$ of $v_g$, the bend point of the edge $(w_j, v_g)$ lies above $\Gamma_{i-1}$; see \Cref{fig:singleton-deg-gt2}. 
Since the height of the flat segment of the edge $(w_j, v_g)$ is $s_j\cdot(x(v_g) - x(w_j))$, where $x(v_g) - x(w_j)$ is its width, the $y$-coordinate of $v_g$ must be at least $H(\Gamma_{i-1})+ s_j \cdot (x(v_g) - x(w_j))$. 
We choose $y(v_g)$ as the smallest multiple of $k$ that fulfills all these bounds. 
Thus, we have
\begin{equation} 
\label{eq:height-2}
y(v_g) \;\le\; H(\Gamma_{i-1}) \;+\; \max\limits_{j \in \{1,\ldots,q-1\}} \{
 s_j\cdot (x(v_g) - x(w_j))\} + k
\;\le\; H(\Gamma_{i-1}) \;+\; 
 W(\Gamma_{i-1}) + k.
\end{equation}

\noindent Having determined the $x$- and $y$-coordinates of $v_g$, the drawing $\Gamma_i$ is completed by drawing each edge incident to $v_g$ in $G_i$ as described above, thereby guaranteeing Invariants~\ref{inv:cut}-\ref{inv:edge-colors}. 
By an argument analogous to that of Case~1, the edges $(v_\ell, v_g)$ and $(v_g, v_r)$ are crossing-free in $\Gamma_i$.
We still have to argue that the edges $(w_j,v_g)$ are also crossing free in $\Gamma_i$. The flat slopes of the edges $(w_j,v_g)$, with $j<q$, have been assigned in counterclockwise order around $v_g$. Also, the left-to-right order of the bend points of these edges matches the left-to-right order of $w_1,\ldots,w_{q-1}$ along the contour $C_{i-1}$; see \Cref{fig:singleton-deg-gt2}. Therefore, no two edges $(w_j,v_g)$ and $(w_{j'},v_g)$ with $0 \leq j < j' \leq q$ can cross each other.
Furthermore, the horizontal segments of the edges $(v_\ell,v_g)$ and $(v_r,v_g)$ lie above the flat segments of the edges $(w_j,v_g)$, while the stretching performed along the edge $(v_\ell,v_{\ell'})$ and $(v_r,v_{r'})$ ensures that the steep segments of $(v_\ell,v_g)$ and $(v_r,v_g)$ lie completely to the left and to the right of these flat segments; see \Cref{fig:badcases}. 
Since the flat segments of the edges $(w_i,v_g)$ are drawn completely above $\Gamma_{i-1}$, they cannot cross any edge of $G_{i-1}$, thereby implying that $\Gamma_i$ is planar, as desired. 
Having ensured this property, the fact that the contour $C_i$ of~$\Gamma_i$ is $x$-monotone is implied by the fact that the slopes of the two segments of the edge $(v_\ell,v_g)$ are in \ref{s:blue}$\;\cup\;$\ref{s:hor}, while the ones of $(v_g,v_r)$ in \ref{s:hor}$\;\cup\;$\ref{s:green} (i.e., symmetrically to Case~1). Hence,  $\Gamma_i$ also satisfies Invariant \ref{inv:xmonotone}.

We now establish an upper bound on how much the drawing $\Gamma_{i-1}$ must be stretched horizontally to accommodate $v_g$ in $\Gamma_i$.
As in Case 1, the edges $(v_{\ell},v_{\ell'})$ and $(v_r,v_{r'})$ have been stretched by at most $\frac{y(v_g)}{s_\ell}$ and $\frac{y(v_g)}{|s_r|}$ units, respectively. Furthermore, the edges $(w_j,w'_j)$ have been stretched by at most $(q-1)\cdot(\Delta-4)$ unit in total. Thus, the total stretch applied is at most
\begin{align}\label{eq:width-2}
&\notag\frac{y(v_g)}{s_\ell} + \frac{y(v_g)}{|s_r|} + (q-1)\cdot (\Delta-4) \\
&\notag\overset{(\ref{eq:height-2})}{\le} \frac{H(\Gamma_{i-1})+ W(\Gamma_{i-1}) + k}{s_\ell}+\frac{H(\Gamma_{i-1}) + W(\Gamma_{i-1}) + k}{|s_r|}+(\deg(v_g)-1)\cdot (\Delta-4)\\
&\notag\le 2\cdot \frac{H(\Gamma_{i-1})+ W(\Gamma_{i-1}) + k}{\frac{k}{ \Delta-3}}+\deg(v_g)\cdot\Delta\\
&\le 2\Delta\cdot \frac{H(\Gamma_{i-1}) + W(\Gamma_{i-1})+ k}{k}+\Delta n.
\end{align}

To complete the description of our drawing algorithm, it remains to consider the case in which $v_n$ is of degree exactly $\Delta$.
Let $w_1,\ldots,w_\Delta$ be the neighbors of $v_n$ as they appear along $C_{m-1}$ from $v_1$ to $v_2$. Before applying our drawing algorithm as described so far, we remove the edge $(w_\Delta,v_n)$ from the graph and recolor the edge $(w_{\Delta-1},v_n)$ green. After the last step of the drawing algorithm, we reinsert the edge $(w_\Delta,v_n)$. Since $w_\Delta$ lies on $C_m$ and its degree is strictly less than $\Delta$ (in the absence of $(w_{\Delta-1},v_n)$), the vertical top port of $w_\Delta$ is unoccupied by Invariant~\ref{inv:ports}.
By Invariant~\ref{inv:xmonotone}, we can draw the edge $(w_{\Delta-1},v_n)$ with a vertical segment at $w_\Delta$ and a segment of slope $\frac{1}{\Delta-3}$ at $v_n$ (see \Cref{fig:1-bend-drawing}). This completes the drawing of the input graph $G$.

It remains to analyze the width and the height of the drawing. For ease of notation, we denote by $W_i$ and $H_i$, $0\le i\le m$ the width and height of the drawing $\Gamma_i$, that is, $W_i=W(\Gamma_i)$ and $H_i=H(\Gamma_i)$. In the base of our recursive algorithm, it holds that $W_0=1$ and $H_0=0$. 
By \Cref{eq:height-2}, we obtain $H_i\le H_{i-1} + W_{i-1} +k$. So, 
\begin{align}
\label{eq:Hl}
H_m\le \sum_{i=0}^{m-1} (W_i+k)\le m\cdot(W_m+k)\le n\cdot W_m+ kn.
\end{align}

\noindent By \Cref{eq:width-1} and \Cref{eq:width-2}, we obtain 
\begin{align*}
W_i&\le W_{i-1}+2\Delta\cdot \frac{H_{i-1}+W_{i-1} + k}{k}+|P_i|+\Delta n\\
&=\frac{2\Delta}{k}\cdot H_{i-1}+(1+\frac{2\Delta}{k})\cdot W_{i-1}+(n+2)\Delta+|P_i|.
\end{align*}

\noindent By series expansion, we obtain
\begin{align*}
\notag W_m &\le \sum_{i=0}^{m-1} \left(1+\frac{2\Delta}{k}\right)^{m-i}\left(\frac{2\Delta}{k} H_i+(n+2)\Delta+|P_i|\right)\\
\notag &\le \left(1+\frac{2\Delta}{k}\right)^{m}\cdot\sum_{i=0}^{m-1} \left(\frac{2\Delta}{k} H_i+(n+2)\Delta+|P_i|\right)\\
&\notag \le \left(1+\frac{2\Delta}{k}\right)^{m}\cdot\left(m\cdot \frac{2\Delta}{k} \cdot H_m+(n+2)\Delta m+n\right) & m\le n, n\ge 6\\
&\le \left(1+\frac{2\Delta}{k}\right)^{n}\cdot\left(n\cdot \frac{2\Delta}{k} \cdot H_m+2\Delta n^2\right).
\end{align*}

\noindent If we choose $k=4\Delta n^2$, then we have $\left(1+\frac{2\Delta}{k}\right)^{n}=\left(1+\frac{1}{2 n^2}\right)^{n}$. Since $\ln(1+x)\le x$ for $0<x<1$, we obtain $\left(1+\frac{1}{2 n^2}\right)^{n}\le e^{1/(2n)} \le e^{1/6} < 1.19$, so
\begin{align}
\label{eq:Wl}
    W_m \le 1.19\left(n\cdot \frac{2\Delta}{k} \cdot H_m+2\Delta n^2\right).
\end{align}

\noindent Plugging \Cref{eq:Wl} into \Cref{eq:Hl}, we obtain
\begin{align*}
    H_m&\le 1.19\cdot\left(n^2\cdot \frac{2\Delta}{k} \cdot H_m+2\Delta n^3\right)+kn
    \le \frac{2\cdot 1.19\cdot k\Delta n^3+k^2 n}{k-2\cdot 1.19\cdot \Delta n^2}  \\
     &= \frac{2\cdot 1.19\cdot 4\Delta^2 n^5+16 \Delta^2 n^5}{4\Delta n^2-2\cdot 1.19\cdot \Delta n^2}  = \frac{25.52\Delta n^3}{1.62}\le 15.76\Delta n^3 \in O(\Delta n^3).
\end{align*}

\noindent  Plugging this back into \Cref{eq:Wl}, we can bound the width by
\begin{align*}
W_m&\le 1.19\left(n\cdot \frac{2\Delta}{4\Delta n^2} \cdot H_m+2\Delta n^2\right) \le \frac{1.19}{2n}\cdot H_m+2.38\Delta n^2 \le \frac{18.76\Delta n^3}{2n}+2.38\Delta n^2 \\
&\le 11.76\Delta n^2\in O(\Delta n^2).
\end{align*}

\noindent Reinserting $(v_1,v_2)$ and $(w_\Delta,v_n)$ at the end increases the height by at most $2W_m/(\Delta-3)\le 11.76\Delta n^2\le 1.96\Delta n^3$, since $\Delta\ge 5$ and $n\ge 6$.
Thus, our drawing has area $O(\Delta^2 n^5)$.
\end{proof}

\noindent Using the algorithms in~\cite{kant1993thesis,KantBodlaender1991} (see \cref{sec:preliminaries}), we can augment any planar graph with maximum degree $\Delta\ge5$ to a triconnected planar graph of maximum degree at most 
$\lceil 3\Delta/2\rceil+3$.

\begin{corollary}\label{th:degree-delta-1-bend2}
    Every planar $n$-vertex graph with maximum degree $\Delta\ge 5$ admits a $1$-bend planar grid drawing with at most $ \lceil \frac{9}{2}\Delta \rceil +1$ slopes on a  $O(\Delta n^2) \times O(\Delta n^3)$ grid.
\end{corollary}

\subsection{1-bend planar drawing of 3-connected degree-5 planar graphs}
For the special case of planar graphs of maximum degree~$5$, we can slightly improve the number of slopes while increasing the required area by a factor of $O(n^2)$. While the previous algorithm utilized $3\Delta - 8$ slopes, we show that for $\Delta = 5$, a set of five slopes is sufficient to have a $1$-bend planar grid drawing. Our construction ensures that all vertices are drawn on grid points and the drawing is planar. 

\begin{theorem} \label{thm:1-bend_5-degree}
        Every $3$-connected planar graph $G$ with maximum degree $\Delta = 5$ admits a $1$-bend planar grid drawing $\Gamma$ using a fixed set of $5$ slopes. Such a drawing can be constructed on a grid of size $O(n^3) \times O(n^4)$.
\end{theorem}

\begin{proof}
    
The slope set $S$ used by our algorithm is defined with respect to a constant $k=5n^2$ as follows:
\begin{enumerate*}[label=S.\arabic*]
    \slopeitem{S_{\mathrm{v}}}{s5:vert}consists only of the vertical slope.
    \slopeitem{S_{\mathrm{h}}}{s5:hor}consists only of the horizontal slope.
    \slopeitem{S_{\mathrm{ls}}}{s5:green}consists only of the \emph{left steep} slope $-k$ (green in~\Cref{fig:degree5}).
    \slopeitem{S_{\mathrm{rs}}}{s5:blue}consists only of the \emph{right steep} slope $k$  (blue in~\Cref{fig:degree5}). Finally,
    \slopeitem{S_{\mathrm{f}}}{s5:red}consists of the \emph{flat} diagonal slope 1
    (red in~\Cref{fig:degree5}).
\end{enumerate*}
We denote $S'=$\;\ref{s5:vert}\;$\cup$\;\ref{s5:green}\;$\cup$\;\ref{s5:blue}.

We replace Invariants~\ref{inv:xmonotone}, \ref{inv:ports} and \ref{inv:edge-colors} with weaker ones. The invariants~\ref{inv:cut} and \ref{inv:ycoord} are preserved. 
\begin{enumerate}[label=I.\arabic*$^\star$]
    \item\label{inv5:xmonotone} The contour of the drawing is not necessarily $x$-monotone. However, all vertices on the contour that still have at least one neighbor in $G\setminus G_{i-1}$ appear in left-to-right order along the contour. Moreover, for each such vertex $v$, the vertical ray corresponding to the slope in \ref{s5:vert} does not intersect any part of the drawing $\Gamma_{i-1}$ (except possibly the first segment of an edge incident to $v$). All non $x$-monotone parts of the contour consist exclusively of vertices that do not have any remaining outgoing edges.

    \setcounter{enumi}{2}
    \item\label{inv5:ports} Every vertex of $C_{i-1}$ has at least as many unoccupied ports in $S'\;$ incident to the outer face of $G_{i-1}$ as it has neighbors in $G \setminus G_{i-1}$.

    \setcounter{enumi}{4}
    \item\label{inv5:edge-colors} Based on their colors, the edges of $G_{i-1}$ have been drawn as follows in $\Gamma_{i-1}$:
    \begin{enumerate}[label=\alph*$^\star$.,ref=I.5.\alph*$^\star$] 
        \item\label{inv5:black} Each black edge of $G_{i-1}$ consists of a single horizontal segment (i.e., its slope is in \ref{s5:hor}).
        \item\label{inv5:blue-green} Each blue and green edge of $G_{i-1}$ consists of two segments; the one incident to its source has a slope in $S'$, while the one incident to its target is in~\ref{s5:hor}. 
        \item \label{inv5:red} Each red edge of $G_{i-1}$ consists of at most two segments; the one incident to its source has a slope in $S'$, while the one incident to its target has a slope in \ref{s5:vert}$\;\cup\;$\ref{s5:red}.  
    \end{enumerate}  
\end{enumerate}

Let $\Pi = (P_0, \ldots, P_m)$ be a canonical ordering of the input graph. As in~\Cref{th:degree-delta-1-bend}, we assume that for some $i>0$ a $1$-bend planar grid drawing $\Gamma_{i-1}$ of $G_{i-1}$ has already been constructed.  The base case, consisting of a single edge, is handled in exactly the same way as before.

We first describe how to handle the case that $P_i=\{v_g\}$ is a singleton of degree greater than~$2$. 
Let $v_\ell,w_1,w_2,v_r$ be the neighbors of $v_g$ in $\Gamma_{i-1}$ as they appear from left to right along $C_{i-1}$ (possibly $w_1=w_2$ if $v_g$ only has degree~3 in $G_i$).
Because of the weaker invariant~\ref{inv5:red}, we can no longer guarantee that the last outgoing edge of a contour vertex uses a vertical slope, it may happen that both $w_1$ and $w_2$ are connected to $v_g$ by sloped segments. 

The edges from $w_1$ and $w_2$ to $v_g$ will be drawn with a segment of any free slope $s_1,s_2\in$\;$S'$ incident to $w_1,w_2$, which is unoccupied by Invariant~\ref{inv5:ports}.
For the edges from $v_\ell$ and $v_r$ to $v_g$, we use the first available slope $s_\ell,s_r\in$\;$S'$ in clockwise order, if the edge is green, or in counter-clockwise order, if the edge is blue. There is at least one unoccupied slope by Invariant~\ref{inv5:ports}, and this choice also ensures that Invariant~\ref{inv5:ports} is maintained.

At $v_g$, we will use a vertical segment for $(w_2,v_g)$, a segment of the slope in \ref{s5:red} for $(w_1,v_g)$ (if $w_1\neq w_2$), and a horizontal segment for $(v_\ell,v_g)$ and $(v_r,v_g)$. This maintains Invariant~\ref{inv5:edge-colors}.

\begin{figure}[t]
\subcaptionbox{\label{fig:degree5-slopes}}{\includegraphics{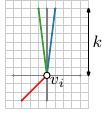}}
\hfill
\subcaptionbox{\label{fig:degree5-drawing}}{\includegraphics[page=2]{figures/degree5.pdf}}
\hfill
\subcaptionbox{\label{fig:nicefigure}}{\includegraphics[page=3]{figures/degree5.pdf}}
\caption{(a) Illustration of the slopes used in \Cref{thm:1-bend_5-degree}. (b) The drawing created by the algorithm in \Cref{thm:1-bend_5-degree} with $k=4$ for the graph in \Cref{fig:schnyder}. (c) Illustration of the case in which $P_i$ is a singleton of degree greater than 2.}
\label{fig:degree5}
\end{figure}

Let $p_1$ and $p_2$ be the points where the rays that correspond to $\rho_1$ and $\rho_2$ intersect the horizontal line $L'$ through $v_{g-1}$ (that is, the top of the drawing $\Gamma_{i-1}$); see \Cref{fig:nicefigure}. We will draw the edges $(w_1,v_g)$ and $(w_2,v_g)$ to have their bend point at $p_1$ and $p_2$, respectively. If $s_1\in$\;\ref{s5:vert}, then we can draw it crossing-free by Invariant~\ref{inv5:xmonotone}. Otherwise, let $(w_1,w'_1)$ be the first edge of $C_{i-1}$ that is encountered when traversing $C_{i-1}$ from $w_1$ to $v_2$ if $s_1\in$\;\ref{s5:blue}, or from $w_1$ to $v_1$ if $s_1\in$\;\ref{s5:green}. To ensure planarity, we stretch this edge such that no part of the drawing lies in the rectangle spanned by $p_1$ and $w_1$, which is possible by Invariant~\ref{inv:cut}. If $s_1\in$\;\ref{s5:blue}, then by Invariant~\ref{inv5:xmonotone}, no vertex or bend point encountered along the traversal of $C_{i-1}$ from $w_1$ to $v_r$ lies to the left of $w_1$. Thus, stretching the edge $(w_1,w'_1)$ by at most $\frac{H(\Gamma_{i-1})}{k}$ ensures planarity. The case $s_1\in$\;\ref{s5:green} is symmetric.
Analogously, we define and stretch the edge $(w_2,w'_2)$.

If $v_g$ has degree 3, we can place it $k$ coordinates above $p_2$. Otherwise,
we have to place $v_g$ at the intersection of the ray through $p_1$ with the slope in \ref{s5:red} and the ray through $p_2$ with the vertical slope. To ensure Invariant~\ref{inv:ycoord}, we might have to stretch the edge $(w_1,w'_1)$ again by up to $k-1$ units such that $x(v_g)-x(p_1)$ and thus $y(v_g)$ becomes a multiple of $k$. 
So far, we have stretched horizontally by at most $2 \cdot \frac{H(\Gamma_{i-1})}{k}+k$, so we have
\begin{align}
    \label{eq5:hor-inner}
    x(v_g)-x(p_1)\le W(\Gamma_{i-1})+2\cdot\frac{H(\Gamma_{i-1})}{k}+k.
\end{align}

\noindent Thus, we obtain 
\begin{align}
    \label{eq5:height}
    \notag y(v_g)&\le H(\Gamma_{i-1})+k+x(v_g)-x(p_1) \le H(\Gamma_{i-1})+k+W(\Gamma_{i-1})+2\cdot\frac{H(\Gamma_{i-1})}{k}+k\\
    &= 2k+W(\Gamma_{i-1})+\left(1+\frac{2}{k}\right)H(\Gamma_{i-1}).
\end{align}

\noindent Let $\rho_\ell$ and $\rho_r$ be the rays starting from $v_l$ and $v_r$ with slope $s_\ell$ and $s_r$, respectively.
Let $p_\ell$ and $p_r$ be the points where $\rho_\ell$ and $\rho_r$ intersect the horizontal line $L$ through $v_g$; see \Cref{fig:nicefigure}. We have to ensure that the segments $(v_\ell,p_\ell)$ and $(v_r,p_r)$ are crossing-free, and that $x(p_\ell)<x(v_g)<x(p_r)$. Defining $(v_\ell,v'_\ell)$ and $(v_r,v'_r)$ similarly to above, we can ensure this by stretching both of them by at most $\frac{y(v_g)}{k}+1$ such that the rectangle spanned by $v_\ell$ and $p_\ell$ as well as the rectangle spanned by $v_r$ and $p_r$ contains no other vertices or bends. This also ensures Invariant~\ref{inv5:xmonotone}.
The total stretching for this step is

\begin{align}
    \label{eq5:hor-outer}
    \notag 2\left(\frac{y(v_g)}{k}+1\right)&\overset{(\ref{eq5:height})}{\le} 2+2\cdot \frac{2k+W(\Gamma_{i-1})+\left(1+\frac{2}{k}\right)H(\Gamma_{i-1})}{k}\\
    & \le 6+\frac{2}{k}W(\Gamma_{i-1})+\left(\frac{2}{k}+\frac{4}{k^2}\right)H(\Gamma_{i-1}).
\end{align}

\noindent The cases of a degree-two singleton and a chain follow the same argument, with the simplification that the connections from $w_1$ and $w_2$ need not be considered.

We will now bound the width and the height of the drawing.
From \Cref{eq5:height}, we get

\begin{align}
    \label{eq5:hm}
    H_m \notag &\le 2k+W_{m-1}+\left(1+\frac{2}{k}\right)H_{m-1} & m\le n\\
    \notag &\le \left( 1 + \frac{2}{k} \right)^n (2kn^2+n W_m) & k=5n^2\\
    \notag &\le \left( 1 + \frac{2}{5n^2} \right)^n (10n^3+n W_m) &  \ln(1+x)\le x\\
    &\le e^{0.4/n} (10n^3+n W_m)
\end{align}

\noindent From \Cref{eq5:hor-inner,eq5:hor-outer}, we get

\begin{align}
    \label{eq5:wm}
    \notag W_m&\le W_{m-1}+2\frac{H_{m-1}}{k}+k+6+\frac{2}{k}W_{m-1}+\left(\frac{2}{k}+\frac{4}{k^2}\right)H_{m-1}& m\le n\\
    \notag &\le \left(1+\frac{2}{k}\right)W_{m-1}+\left(\frac{4}{k}+\frac{4}{k^2}\right)H_{m-1}+k+6\\
    \notag &\le \left(1+\frac{2}{k}\right)^n\left(\left(\frac{4}{k}+\frac{4}{k^2}\right)nH_{m}+kn+6n\right) & k=5n^2 \\
     & \le e^{0.4/n} \left(\left( \frac{4}{5n^2} + \frac{4}{25n^4}\right)nH_m +5n^3 +6n \right)
\end{align}

\noindent Plugging \Cref{eq5:hm} into \Cref{eq5:wm}, we obtain
\begin{align}
\label{eq5:wmmmmmmmm}
    \notag W_m & \le e^{0.8/n} n (10n^3 + nW_m) \left(\frac{4}{5n^2} + \frac{4}{25n^4} \right) + 5n^3 +6n \\
    \notag & \le 10 e^{0.8/n} n^4  \left(\frac{4}{5n^2} + \frac{4}{25n^4} \right) + e^{0.8/n} n^2  \left(\frac{4}{5n^2} + \frac{4}{25n^4} \right) W_m +5n^3 + 6n \\
    \notag & \le \frac{10e^{0.8/n}n^5\left(\frac{4}{5n^3}+\frac{4}{25n^4}\right)+5n^4+6n}{1-\frac{4n^2e^{0.8/n}}{5n^3}-\frac{4n^2e^{0.8/n}}{25n^4}} \\
    \notag & \le \frac{125n^5 + 200 e^{0.8/n} n^4 + 150 n^3 + 40 e^{0.8/n} n^2}{25n^2 - 20 e^{0.8/n} n^2 - 4e^{0.8/n}} & \hspace{-2cm}e^{0.8/n}<1.15\text{ for }n\ge 6 \\
    & < \frac{62.5n^5 + 115n^4+75n^3+23n^2}{n^2 -2.3} 
    \le 67 n^3 + O(n^2)
\end{align}

\noindent Plugging \Cref{eq5:wmmmmmmmm} back into \Cref{eq5:hm}, we can bound the height by

\begin{align*}
    \notag H_m &\le e^{0.4/n} \left(10n^3+n (67n^3 + O(n^2)\right) & \hspace{3.8cm}e^{0.4/n}<1.07\text{ for }n\ge 6\\
    & < 71.69n^4 + 10.7n^3 + O(n^2).
\end{align*}

\noindent Reinserting the edges $(v_1,v_2)$ and the missing edge of $v_n$ at the end increases the height by at most $2W_m\in O(n^3)$ and the width by at most $2H_m/k\in O(n^2)$. We therefore obtain the following theorem.
\end{proof}

\section{\texorpdfstring{2-bend planar drawings of degree-$\Delta$ graphs}{2-bend planar drawings of degree-Δ graphs}}
\label{sec:2-bend}

In this section, we extend a result of Keszegh, Pach, and Pálvölgyi~\cite{KPP2010GD,KPP2013} by showing that every planar graph $G$ with maximum degree $\Delta \geq 3$ admits a 2-bend planar drawing on a grid of polynomial size using at most $\lceil \Delta/2 \rceil$ slopes, if $G$ is biconnected and at most $\lceil \Delta/2 \rceil+1$ slopes otherwise.
We first address the former case.

\begin{theorem}
	\label{thm:biedl-kant-biconnected}
	Every biconnected planar graph $G$ with maximum degree $\Delta \geq 3$ admits a $2$-bend planar drawing with at most $\lceil \Delta/2 \rceil$ distinct slopes on a grid of size $O(n) \times O(n^2 \Delta)$. The only exception is the octahedron graph, which requires $3$ slopes.
\end{theorem}

\begin{proof}
    Without loss of generality, we may assume that $\Delta$ is even. Since planar graphs with $\Delta=4$ always admit an orthogonal drawing with at most two bends per edge on a grid of size $O(n) \times O(n)$~\cite{BiedlKant1998, PapakostasTollis1996}, we may further assume $\Delta \geq 6$. 
    We choose a vertex $v_n$ of $G$ with degree at most~$5$ (which exists in every planar graph) and consider a planar embedding $\cal E$ of $G$ with $v_n$ on its outer face. 
    The fact that $G$ is biconnected implies that it admits an $st$-ordering $v_1,\ldots,v_n$, such that $(v_1,v_2)$ is an edge of the outer face of~$\mathcal{E}$.
    We will use this $st$-ordering to construct incrementally a $2$-bend planar drawing $\Gamma$ of $G$ on the following set of $\frac{\Delta}{2}$ slopes: $S = \{-\lfloor\frac{\Delta}{4}\rfloor+1,\ldots,0,\ldots,\lceil\frac{\Delta}{4}\rceil-1,\infty\}$; see \Cref{fig:slopes_biedl_kant_1}. 
    In particular, each edge $(v_i,v_j)$ with $i<j$ will be drawn with a \emph{first non-vertical segment} incident to $v_i$ (of possibly zero length), followed by a vertical segment (of non-zero length) and a \emph{last non-vertical segment} incident to $v_j$ (of possibly zero length); see \Cref{fig:slopes_biedl_kant_2}.

	\begin{figure}[tb]
		\captionsetup[subfigure]{justification=centering}
		\subcaptionbox{\label{fig:slopes_biedl_kant_1}}{
			\includegraphics[page=1]{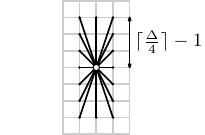}}
        \hfil
		\subcaptionbox{\label{fig:slopes_biedl_kant_2}}{
			\includegraphics[page=2]{slopes_biedl_kant}}
        \hfil
		\subcaptionbox{\label{fig:slopes_biedl_kant_3}}{
			\includegraphics[page=3]{slopes_biedl_kant}}
		\caption{(a) The set of $\lceil \Delta/2 \rceil$ slopes used in~\Cref{thm:biedl-kant-biconnected}. (b) Example construction of a $2$-bend edge $e=(v_i,v_j)$. (c) A sketch of our drawing.}
		\label{fig:slopes_biedl_kant}
	\end{figure}
    
    We first compute a \emph{sketch} drawing $\Gamma'$ of $G$ using the algorithm of \cite{KPP2010GD, KPP2013} and taking $v_1,\ldots,v_n$ as the input $st$-ordering. 
    In the resulting drawing $\Gamma'$, each edge is drawn with at most two bends and contains a vertical segment. These vertical segments lie in consecutive columns of $\Gamma'$ that also contain the vertices of $G$. 
    Given a vertex $v$ of $G$, we denote by $x(v)$ the index of the column of $\Gamma'$ that contains $v$ in the left-to-right ordering of the columns of~$\Gamma'$. 
    Accordingly, given an edge $e$, with slight abuse of notation, we denote by $x(e)$ the index of the column of $\Gamma'$ that contains the vertical segment of $e$ in the left-to-right ordering of the columns of $\Gamma'$; see \Cref{fig:slopes_biedl_kant_1}.
    Let $S'=\{s'_{-\lfloor\Delta/4\rfloor+1},\ldots,s'_0,\ldots,s'_{\lceil\Delta/4\rceil-1},s'_\infty\}$ be the slopes used in~$\Gamma'$ such that $s'_\infty$ is the vertical slope and $s'_i<s'_j$ if $i<j$.
    We produce a drawing~$\Gamma$ of $G$ on a grid of size $O(n) \times O(n^2 \Delta)$ such that every edge segment that is drawn with slope $s'_i$ in~$\Gamma'$ is drawn with slope $i$ in $\Gamma$.    
    Hence, $\Gamma$ and $\Gamma'$ share the same planar embedding. 
    Furthermore, the $x$-coordinate of each vertex $v$ of $G$ will be $x(v)$ in~$\Gamma$, while the $x$-coordinate of the vertical segment of each edge $e$ of $G$ will be $x(e)$. Hence, the width of~$\Gamma$ will be $O(n)$.
    
    Let $G_i$ with $1 \leq i \leq n$ be the subgraph of $G$ induced by $\{v_1, \dots, v_i\}$, and let $\mathcal{E}_i$ be the restriction of $\mathcal{E}$ to $G_i$. We say that an edge $(v_j, v_k)$ of $G$ is a \emph{pending} edge of $G_i$ if and only if $j \le i < k$.  
    For $2 \leq i \leq n$, we denote by $\Gamma_i$ a drawing of $G_i$ which additionally contains the first non-vertical segment (if any) of each pending edge of $G_i$. 
    In the following, we will first describe how to compute drawing $\Gamma_2$. Then, assuming that we have recursively computed a drawing $\Gamma_{i-1}$ with $2 \leq i< n$, we describe how to compute drawing $\Gamma_i$.

    To obtain drawing $\Gamma_2$, we begin by placing $v_1$ at $(0,0)$ and $v_2$ at $(x(v_2), 0)$, as in the original algorithm. We connect $v_1$ to~$v_2$ with a $2$-bend edge drawn below both vertices. This edge consists of two vertical segments, each of length $\lfloor \frac{\Delta}{4} \rfloor (2m - n)+1$, incident to  $v_1$ and $v_2$, and a horizontal segment connecting their lower endpoints; see \Cref{fig:slopes_biedl_kant_3}.
    To complete drawing $\Gamma_2$, we next draw the first non-vertical segment of each pending edge that is incident to $v_1$ and $v_2$ preserving $\mathcal{E}_2$ as follows. For $j\in\{1,2\}$, consider a pending edge $e=(v_j,v_k)$ with $k>2$ and suppose that this edge has a non-vertical segment incident to $v_j$ in $\Gamma'$. Let $s'_\ell$ be the slope of this segment in $\Gamma'$. We draw a segment starting at $v_j$ with slope~$\ell$ until its $x$-coordinate equals $x(e)$. This point serves as the first bend of $e$ and is a grid point, since $S$ consists of integer slopes. Furthermore, the vertical extent of this segment is at most $\lfloor \Delta/4\rfloor \cdot |x(e) - x(v_j)|$. This guarantees that this segment does not cross $(v_1,v_2)$, since the width of the drawing is bounded by $2m-n$, as we will shortly show, and thus, $|x(e) - x(v_j)| \leq 2m-n$.
        
    We now describe how to compute drawing $\Gamma_i$, assuming that we have recursively constructed drawing $\Gamma_{i-1}$.
    We place $v_i$ at $x$-coordinate $x(v_i)$ and $y$-coordinate $(i-1) \cdot (\Delta/2 \cdot (2m - n)+1)$, such that $v_i$ lies above one of its median predecessors in $G_i$, as in the original algorithm. 
    Then, we complete the drawing of the edges connecting $v_i$ to its neighbors in $G_{i-1}$ while preserving $\mathcal{E}_i$. Consider such an edge $e=(v_j,v_i)$ with $i>j$ and assume that $s'_\ell$ is the slope of its last non-vertical segment in $\Gamma'$. In drawing $\Gamma_i$, we draw a segment starting at $v_i$ with slope $\ell$ until its $x$-coordinate equals $x(e)$ followed by a vertical segment that connects it to the endpoint of the first non-vertical segment of the edge $e$ (which has already been drawn, when $\Gamma_j$ was computed).    
    To complete the drawing of $\Gamma_i$, we draw the first non-vertical segments of the pending edges incident to $v_i$ in the same way as described for $\Gamma_2$. 
    It follows that the vertical extent of any non-vertical segment incident to $v_i$ belonging to edge $e$ is bounded by $\lfloor \Delta/4\rfloor \cdot |x(e) - x(v_i)|$.
    
    Since each vertex $v_i$ and its incident edges occupy at most $\deg(v_i)-1$ columns, the total width of the resulting drawing $\Gamma$ of $G$ is bounded by $\sum_{i=1}^n (\deg(v_i)-1) = 2m - n\in O(n)$.
    To estimate the height of $\Gamma$, let $H_{i}$ be the maximum $y$-coordinate in $\Gamma_{i}$. 
    For each $i > 2$, $H_{i}$ is determined by the $y$-coordinate of $v_{i}$ plus the vertical extent of the first non-vertical segment of each of the pending edges of $v_i$ plus the vertical extent of edge $(v_1, v_2)$. 
    Thus, 
    \begin{align*}
    	\notag H_{i-1} &\leq y(v_{i-1}) + \lfloor\frac{\Delta}{4}\rfloor \cdot W_{i-2} + \lfloor \frac{\Delta}{4} \rfloor (2m - n)+1 \\
        &\le  (i-2) \cdot \left(\frac{\Delta}{2} \cdot (2m - n)+1\right) +  \frac{1}{2} \cdot \frac{\Delta}{2} \cdot (2m - n) + \lfloor \frac{\Delta}{4} \rfloor (2m - n)+1 \\
        &< \left(\frac{\Delta}{2} \cdot (2m-n)+1\right)\left(i - \frac{3}{2}\right) + \lfloor \frac{\Delta}{4} \rfloor (2m - n)+1.
    \end{align*}
    The lowest point of any segment of a pending edge of $v_i$ in $\Gamma_i$ is at least 
    \begin{equation*}
    	y(v_i) - \lfloor\frac{\Delta}{4}\rfloor \cdot W_{i-1} + \lfloor \frac{\Delta}{4} \rfloor (2m - n)+1 > \left(\frac{\Delta}{2} \cdot (2m-n)+1\right)\left(i - \frac{3}{2}\right) + \lfloor \frac{\Delta}{4} \rfloor (2m - n)+1 .
    \end{equation*}
    This ensures that all segments of the pending edge of $v_i$ lie above drawing $\Gamma_{i-1}$. Hence $\Gamma_i$ is planar.
  	The total height of the drawing is $O(n \cdot \Delta \cdot n) = O(\Delta n^2)$.
\end{proof}

Since any simply connected planar graph can be augmented to a biconnected planar graph by adding auxiliary edges such that every vertex receives at most two augmenting incident edges \cite{kant1993thesis, KantBodlaender1991}, we obtain the following corollary.

\begin{corollary}
	\label{cor:biedl-kant-simply}
	Every planar graph $G$ with maximum degree $\Delta$ admits a $2$-bend planar drawing with at most $\lceil \Delta/2 \rceil + 1$ distinct slopes on a grid of size $O(n) \times O(\Delta n^2)$.
\end{corollary}

\section{\texorpdfstring{4-bend planar drawings of degree-$\Delta$ graphs}{4-bend planar drawings of degree-Δ graphs}} \label{sec:4-bend}

In this section, we seek to prove that every planar graph of maximum degree $\Delta$ admits a $4$-bend planar grid drawing with at most $\Delta$ slopes on an $O(n)\times O(n)$ grid. Our approach is based on an algorithm by Kaufmann and Wiese~\cite{KaufmannWiese2002}, which given a planar graph $G=(V,E)$ and a set of points $P$ in the plane such that no two points have the same $x$-coordinate, it computes a $2$-bend  planar drawing of $G$ that maps each vertex in $V$ to a point in $P$.
Let the \emph{spine} be the $x$-monotone polyline whose bend-points are exactly the points in $P$. In the produced drawing, every edge is one of the following  (see \Cref{fig:kaufmannwiese}):
\begin{enumerate*}[label=\bf (\roman*), ref=(\roman*)]
    \item a \emph{top edge}: it is drawn completely above the spine with one bend;
    \item a \emph{bottom edge}: it is drawn completely below the spine with one bend; or
    \item a \emph{spine-crossing edge}: it crosses the spine exactly once and has two bend points: one below and one above the spine.
\end{enumerate*}
Note that the drawings produced by their algorithm require a linear number of slopes and, as stated in \cite[Lemma 3.2]{KaufmannWiese2002}, exponential area (in the number of vertices).

\begin{theorem}
\label{thm:kaufmann-wiese}
    Every planar graph of maximum degree $\Delta$ admits a $4$-bend planar drawing with $\Delta$ slopes on a grid of size $O(n)\times O(n)$.
\end{theorem}
\begin{proof}
We will create a drawing that is conceptually similar to the one by Kaufmann and Wiese. 
We first apply their algorithm on a set of points that lie on the diagonal $y=x$ to obtain a drawing $\Gamma$ of $G$. 
Let $\Gamma'$ be the drawing of the graph $G'=(V',E')$ obtained from~$\Gamma$ by subdividing every spine-crossing edge $e$ by a dummy vertex placed on the crossing between~$e$ and the spine. 
We will show how to obtain a $3$-bend planar drawing with $\Delta$ slopes for $G'$ on a grid of size $O(n)\times O(n)$, in which 
the two edges incident to each dummy vertex use at most two bends and are drawn such that they use opposite (horizontal or vertical) ports at the incident dummy vertex. 
Thus, smoothing the dummy vertices yields a $4$-bend drawing of $G$ with $\Delta$ slopes on a grid of size $O(n)\times O(n)$.

Let $v_1,\ldots, v_{n'}$ be the vertices of $V'$ in the order that they appear along  the spine in $\Gamma'$ and assume that every edge $(v_i,v_j)$ with $i<j$ is oriented from $v_i$ to $v_j$. 
Let $t_i$ and $b_i$ be the number of top and bottom edges of $v_i$, respectively.
We place $v_1$ at point $(0,0)$ and every other vertex $v_i$ with $1 < i\leq n'$ at point $p_{v_i}=p_{v_{i-1}}+(d_{i-1},d_{i-1})$, where $d_{i-1}= \max\{1,t_{i-1},b_{i}\}$. Furthermore, the $j$-th top (bottom, respective) edge incident to $v_i$ in clockwise order around $v_i$ starting from the spine will have a bend at point at $p_{v_i}+\mathbf{v}_j$ ($p_{v_i}-\mathbf{v}_j$, respectively), where $\mathbf{v}_j=(j-2,j-1)$. This ensures that $\Delta$ distinct slopes are used by the edge segments incident to each vertex; all the edge segments not directly incident to a vertex are drawn either horizontally or vertically. 

To achieve the above properties, we process the vertices in the order $v_1,\ldots,v_{n'}$. Assume that we have already processed the vertices $v_1,\ldots,v_{i-1}$, and consider the next vertex~$v_i$. We call an edge $(v_k,v_\ell)$ with $k<\ell$ \emph{open} if we have already processed $v_k$ but not yet $v_\ell$, that is, $k < i \leq \ell$. As an invariant of our algorithm, each open top edge is assigned to a grid column that will contain its vertical segment, and each open bottom edge is assigned to a grid row that will contain its horizontal segment; see \Cref{fig:kaufmann-wiese-drawing-intermediate}.
 
Let $e_1^{\mathrm t},\ldots,e_{t_i}^{\mathrm t}$ denote the top edges of $v_i$ in clockwise order around $v_i$, starting from the spine.
As already mentioned, for each such edge $e_j^{\mathrm t}$, we place a bend at  point $p_{v_i}+\mathbf{v}_j$.  
If $e_j^{\mathrm t}$ is open, we assign it to the grid column $x(v_i)+j-2$.
Otherwise, we draw a horizontal segment to the left until reaching the grid column previously assigned to $e_j^{\mathrm t}$, and then a vertical segment downward to the bend point located in the neighborhood of the other endpoint of $e_j^{\mathrm t}$.  
In this way, each edge $e_j^{\mathrm t}$ is drawn using at most three bends.
Symmetrically, let $e_1^\mathrm{b},\ldots,e_{b_i}^\mathrm{b}$ be the bottom edges of $v_i$ in clockwise order around $v_i$, starting from the spine.
For each such edge $e_j^\mathrm{b}$, we place a bend at the point $p_{v_i}-\mathbf{v}_j$. 
If $e_j^{\mathrm b}$ is open, we assign it to the grid row $y(v_i)-j+1$.
Otherwise, we draw a vertical segment to the bottom and then a horizontal segment leftward to the bend point located in the neighborhood of the other endpoint of~$e_j^{\mathrm b}$. Hence, each edge $e_j^{\mathrm b}$ is also drawn using at most three bends and the invariant of our algorithm is satisfied.

Once all vertices have been processed, he have obtained a $3$-bend drawing of $G'$; see \Cref{fig:kaufmann-wiese-drawing-final} for an illustration. The obtained drawing is planar because (i)~no two edge segments are assigned to the same horizontal or vertical grid column, and (ii)~throughout the incremental drawing construction, the planar embedding of the algorithm by Kaufmann and Wiese is maintained due to the choice of the bend points around each vertex.

\begin{figure}
    \centering
    \subcaptionbox{\label{fig:kaufmannwiese}}{\includegraphics[page=2]{kaufmannwiese}}
    \hfill
    \subcaptionbox{\label{fig:slopes_kaufmannwiese}}{\includegraphics{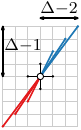}}
    \hfill
    \subcaptionbox{\label{fig:kaufmann-wiese-drawing-intermediate}}{\includegraphics[page=3]{kaufmannwiese}}
    \hfill
    \subcaptionbox{\label{fig:kaufmann-wiese-drawing-final}}{\includegraphics[page=4]{kaufmannwiese}}
\caption{Illustration for our algorithm in \Cref{thm:kaufmann-wiese} for the graph in \Cref{fig:kaufmann-wiese-drawing}. (a) A $2$-bend planar drawing obtained by the algorithm of Kaufmann and Wiese~\cite{KaufmannWiese2002}. Top edges are drawn blue, bottom edges red, and spine-crossing edges green. (b) The vectors used around each vertex in \Cref{thm:kaufmann-wiese}. (c) The drawing after processing $v_3$ and (d) the final drawing.}
    \label{fig:kaufmann-wiese-drawing}
\end{figure}

We now argue that after smoothing the dummy vertices, we obtain a $4$-bend drawing of $G$.
If $v_i$ is a dummy vertex of $G'$, then it has exactly two edges: one incoming edge $(v_h,v_i)$ and one outgoing edge $(v_i,v_j)$ with $h<i<j$, one of which is a top edge, while the other one is a bottom edge. If $(v_h,v_i)$ is a top edge, then both edges use a horizontal segment around~$v_i$. Hence, after smoothing $v_i$, the spine-crossing edge $(v_h,v_j)$ will be drawn with a short segment at $v_h$, followed by a long vertical segment, a horizontal segment that crosses the spine at the old position of $v_i$, a long vertical segment, and a short segment at $v_j$, so it has four bends in total. The case that $(v_h,v_i)$ is a bottom edge is symmetric.

For the drawing area, recall that the difference in $x$-coordinates (and $y$-coordinates) between two consecutive vertices $v_i$ and $v_{i+1}$ is $d_i=\max\{1,t_{i},b_{i+1}\}\le 1+t_i+b_{i+1}$, so we have $x(v_{n'})=y(v_{n'})\le n'+\sum_{i=1}^{n'-1}t_i+\sum_{i=2}^{n'}b_i = n'+2m-t_{n'}-b_1$.
Some edges can extend to the bottom-left of $v_1$ and to the top-right of $v_{n'}$. Namely, one top edge might occupy one column to the left of $v_1$, $b_1-2$ bottom edges might occupy one column to the left of $v_1$ each, and $b_1-1$ bottom edges might occupy one row to the bottom of $v_1$ each. Furthermore, to the right of $v_{n'}$, one column might be occupied by a bottom edge, and $t_{n'}-2$ columns might be occupied by a top edge; above $v_{n'}$, $t_{n'}-1$ rows might be occupied by top edges. Thus, the total width and height of the drawing is at most 
\begin{align*}
x(v_{n'})-x(v_1)+b_1+1+t_{n'}+1&\le n'+2m-t_{n'}-b_1-0+b_1+t_{n'}+2\\
&=n'+2m+2\le (n+m)+2m+2\le 10n-16\in O(n).\qedhere
\end{align*}
\end{proof}

\noindent Since every \emph{subhamiltonian} graph (i.e., a subgraph of a planar Hamiltonian graph) admits an embedding consisting exclusively of top and bottom edges, the following is a direct consequence of~\Cref{thm:kaufmann-wiese}.

\begin{corollary}\label{cor:3-bend}
Every subhamiltonian graph of maximum degree $\Delta$ admits a $3$-bend planar drawing with $\Delta$ slopes on a grid of size $O(n)\times O(n)$.    
\end{corollary}

\section{Conclusions and Open Problems}\label{sec:conclusions}

In this work, we studied trade-offs between the number of slopes, the number of bends per edge, and the area requirements in planar graph drawing. Our results show that allowing only a small number of bends per edge suffices to obtain polynomial-area drawings while maintaining a relatively small slope set. Our results narrow a gap between previous approaches that achieved few slopes at the expense of super-polynomial area and those that focused primarily on low-degree graphs. However, several questions remain open. Besides tightening the trade-offs between slopes and bends and extending the study to non-planar graph classes, we identify the following open problems.

\begin{itemize}
    \item The algorithm presented in \Cref{sec:1-bend} yields drawings in which almost all edges have a bend. The only exception are the black edges and a red incoming edge per vertex. Following an argument of~\cite{KindermannMSS19}, we can choose a canonical order such that there are at most $(2n+1)/3$ vertices without an incoming red edge. This gives us an upper bound of $(8n-17)/3$ on the total number of bends for the constructed drawing. Adjusting the drawing algorithms, so as to reduce the total number of bends in the resulting $1$-bend drawing is an interesting open problem for future consideration. 
    \item We were unable to derive an extension of the algorithm presented in \Cref{sec:1-bend} to handle biconnected (and simply connected) planar graphs without increasing the number of slopes. Such an extension appears to be non-trivial and may require additional properties on the produced drawings in order to handle rigid 3-connected components of the graph.   
    \item The drawings produced by our algorithms have low \emph{angular resolution}, that is, the minimum angle between any two edge segments incident to the same vertex is small. Developing techniques that also account for angular resolution is an interesting direction for future~work.
    \item It is known that the straight-line slope number of outerplanar graphs and partial 2-trees is $\Delta-1$ and $2\Delta$, respectively~\cite{KnauerMW14,LenhartLMN23}. However, the existing algorithms for these results require superpolynomial area. A natural question is whether we can achieve polynomial area drawings for these graph classes by allowing a slight increase in the number of slopes.
\end{itemize}

\bibliographystyle{plainurl}
\bibliography{bibliography}

\end{document}